\documentclass[letterpaper, 10 pt, conference]{ieeeconf}

\usepackage[usenames,dvipsnames,svgnames,table]{xcolor}
\usepackage{amssymb,latexsym,amsfonts,amsmath}
\usepackage{graphicx}
\usepackage{eso-pic}
\usepackage{epsfig}
\usepackage{mathtools}
\usepackage{tikz}
\usetikzlibrary{automata}
\usetikzlibrary{shapes}
\usetikzlibrary{calc,shapes,arrows}
\usepackage{mdwlist}
\usepackage{refcount}
\usepackage{comment}

\usepackage{keyval,trig}
\usepackage{amssymb,dsfont}
\usepackage{mathtools,mathrsfs}

\newtheorem{theorem}{Theorem}[section]

\newtheorem{corollary}[theorem]{Corollary}

\newtheorem{definition}[theorem]{Definition}

\newtheorem{remark}[theorem]{Remark}
\newtheorem{assumption}{Assumption}

\xdefinecolor{MATLABblue}{rgb}{0.0,0.45,.74}
\xdefinecolor{MATLABred}{rgb}{0.85,0.33,.1}

\newcommand{\R}{{\mathbb{R}}}

\newcommand{\N}{{\mathbb{N}}}

\newcommand{\intcc}[1]{\ensuremath{{\left[#1\right]}}}

\newcommand{\as}{\overset{\ra}{s}}

\IEEEoverridecommandlockouts

\usepackage[normalem]{ulem}

\newcommand{\ra}{\rightarrow}

\newcommand{\Let}{:=}
\newcommand{\EE}{\mathds{E}}
\newcommand{\PP}{\mathds{P}}

\title{Compositional Abstractions of Interconnected Discrete-Time\\ Stochastic Control Systems
\thanks{This work was supported in part by the German Research Foundation
(DFG) through the grant ZA 873/1-1.}}

\author{Abolfazl Lavaei, Sadegh Esmaeil Zadeh Soudjani, Rupak Majumdar, and Majid Zamani	
\thanks{A. Lavaei and M. Zamani are
	with the Department of Electrical and Computer Engineering, Technical University
	of Munich, D-80290 Munich, Germany. S. Esmaeil Zadeh Soudjani and R. Majumdar are with the Max Planck Institute for Software Systems, Kaiserslautern 67663, Germany. Email:
	\texttt{\{lavaei,zamani\}@tum.de}, \texttt{\{sadegh,rupak\}@mpi-sws.org}.}
}

\begin{document}
\maketitle

\begin{abstract}
This paper is concerned with a compositional approach for constructing abstractions of interconnected discrete-time stochastic control systems. The abstraction framework is based on new notions of so-called stochastic simulation functions, using which one can quantify the distance between original interconnected stochastic control systems and their abstractions in the probabilistic setting. Accordingly, one can leverage the proposed results to perform analysis and synthesis over abstract interconnected systems, and then carry the results over concrete ones. In the first part of the paper, we derive sufficient small-gain type conditions for the compositional quantification of the distance in probability between the interconnection of stochastic control subsystems and that of their abstractions. In the second part of the paper, we focus on the class of discrete-time linear stochastic control systems with independent noises in the abstract and concrete subsystems. For this class of systems, we propose a computational scheme to construct abstractions together with their corresponding stochastic simulation functions. We demonstrate the effectiveness of the proposed results by constructing an abstraction (totally 4 dimensions) of the interconnection of four discrete-time linear stochastic control subsystems (together 100 dimensions) in a compositional fashion.
\end{abstract}

\section{Introduction}
Large-scale interconnected systems have received significant attentions in the last few years due to their presence in real life systems including power networks, air traffic control, and so on. Each complex real-world system can be regarded as an interconnected system composed of several subsystems. Since these large-scale network of systems are inherently difficult to analyze and control, one can develop compositional schemes to employ the abstractions of the given systems as a replacement in the controller design process. In other words, in order to overcome the computational complexity in large-scale interconnected systems, one can abstract the original concrete system by a simpler one with lower dimension. Those abstractions allow us to design controllers for them, and then refine the controllers to the ones for the concrete complex systems, while provide us with the quantified errors in this controller synthesis detour.

In the past few years, there have been several results on the construction of (in)finite abstractions for stochastic systems.
Existing results include infinite approximations for a class of stochastic hybrid systems \cite{julius2009approximations} and
finite approximations for discrete-time stochastic models with continuous state spaces \cite{tkachev2011infinite,SA13,SSoudjani}.
%approximate similarity relation for controller synthesis over general Markov decision processes with uncountably-infinite state spaces and metric output spaces \cite{},
%infinite-horizon properties over ,
Construction of finite bisimilar abstractions for stochastic control systems is proposed in \cite{zamani2014symbolic,ZTA1}.
%and finally randomly switched stochastic systems \cite{zamani2014approximately}.
Recent results address stochastic switched systems \cite{zamani2014approximately,zamani2015symbolic} and propose compositional construction of infinite abstractions of continuous-time stochastic control systems \cite{zamani2014compositional,zamani2016approximations} using small-gain type compositional reasoning.

In this paper, we provide a compositional approach for the construction of infinite abstractions of interconnected discrete-time stochastic control systems. Our abstraction framework is based on a new notion of so-called stochastic simulation functions under which an abstraction, which is itself a discrete-time stochastic control system with lower dimension, performs as a substitute in the controller design process. The stochastic simulation function is leveraged to quantify the error in probability in this controller synthesis scheme.  As a consequence, one can use the proposed results here to solve particularly safety/reachability problems over the abstract interconnected systems and then carry the results over the concrete interconnected ones. It should be noted that the existing compositional results in \cite{zamani2014compositional, zamani2016approximations} are for \emph{continuous-time} stochastic systems and assume that the noises in the concrete and abstract systems are the same, which means the abstraction has access to the noise of the concrete system, which is a strong assumption. In this paper, we do not have such an assumption meaning that the noises of the abstraction can be completely independent of that of the concrete system. 

%The rest of the paper is organized as follows. Section II provides some mathematical preliminaries and notations, and also the definition of discrete-time stochastic control systems. In Section III, we first introduce a notion of so-called pseudo-simulation functions for the discrete-time stochastic control subsystems with internal inputs; and then define the stochastic simulation functions for the interconnected systems. Section IV contains the main contribution of the paper, that is compositional abstractions of interconnected discrete-time stochastic control systems, by deriving sufficient small-gain type conditions. In Section V, we focus on the specific class of discrete-time linear stochastic control systems, and propose a computational scheme to construct abstractions together with their stochastic pseudo-simulation functions. In the last section, to show the effectiveness of the proposed results, we construct an abstraction (totally 4 dimensions) of the interconnection of four discrete-time linear stochastic control subsystems (totally 100 dimensions) in a compositional fashion.

\section{Discrete-Time Stochastic Control Systems}

\subsection{Preliminaries}
We consider a probability space $(\Omega,\mathcal F_{\Omega},\PP_{\Omega})$,
where $\Omega$ is the sample space,
$\mathcal F_{\Omega}$ is a sigma-algebra on $\Omega$ comprising subsets of $\Omega$ as events,
and $\PP_{\Omega}$ is a probability measure that assigns probabilities to events.
We assume that random variables introduced in this article are measurable functions of the form $X:(\Omega,\mathcal F_{\Omega})\rightarrow (S_X,\mathcal F_X)$.
Any random variable $X$ induces a probability measure on  its space $(S_X,\mathcal F_X)$ as $Prob\{A\} = \PP_{\Omega}\{X^{-1}(A)\}$ for any $A\in \mathcal F_X$.
We often directly discuss the probability measure on $(S_X,\mathcal F_X)$ without explicitly mentioning the underlying probability space and the function $X$ itself.

A topological space $S$ is called a Borel space if it is homeomorphic to a Borel subset of a Polish space (i.e., a separable and completely metrizable space).
Examples of a Borel space are the Euclidean spaces $\mathbb R^n$, its Borel subsets endowed with a subspace topology, as well as hybrid spaces.
Any Borel space $S$ is assumed to be endowed with a Borel sigma-algebra, which is
denoted by $\mathcal B(S)$. We say that a map $f : S\rightarrow Y$ is measurable whenever it is Borel measurable.

%\textcolor{red}{Sadegh: Here we don't need to define the filtration, Brownian motion and Poisson process. We are working with discrete-time processes and need to construct the product space. The set $\mathbb R_{\ge 0}$ should be replaced by $\mathbb N$.}

\subsection{Notation}
The following notation is used throughout the paper. We denote the set of nonnegative integers by $\mathbb N := \{0,1,2,\ldots\}$ and the set of positive integers by $\mathbb Z_+ := \{1,2,3,\ldots\}$.
%The bounded set of integers is indicated by $\mathbb N[a,b]:=\{a,a+1,\ldots,b\}$ for any $a,b\in\mathbb N,\,a\le b$.
%We denote with $\mathbb I(\cdot)$ the indicator function which takes a Boolean-valued expression as an argument and gives $1$ if this expression evaluates to \textsf{true} and 0 when it is \textsf{false}.
The symbols $\R$, $\R_{>0}$, and $\R_{\ge 0}$ denote the set of real, positive, and nonnegative real numbers, respectively. Given a vector $x\in\mathbb{R}^{n}$, $\Vert x\Vert$ denotes the Euclidean norm of $x$. The symbols $I_n$ and $\mathbf{1}_{n}$ denote the identity matrix in $\R^{n\times{n}}$, and the vector in $\R^n$ with all its elements to be one, respectively. We denote by $\mathsf{diag}(a_1,\ldots,a_N)$ a diagonal matrix in $\R^{N\times{N}}$ with diagonal matrix entries $a_1,\ldots,a_N$ starting from the upper left corner. Given functions $f_i:X_i\rightarrow Y_i$,
for any $i\in\mathbb\{1,\ldots,N\}$, their Cartesian product $\prod_{i=1}^{N}f_i:\prod_{i=1}^{N}X_i\rightarrow\prod_{i=1}^{N}Y_i$ is defined as $(\prod_{i=1}^{N}f_i)(x_1,\ldots,x_N)=[f_1(x_1);\ldots;f_N(x_N)]$.
For any set $A$ we denote by $A^{\mathbb N}$ the Cartesian product of a countable number of copies of $A$, i.e., $A^{\mathbb N} = \prod_{k=0}^{\infty} A$.
Given a measurable function $f:\mathbb N\rightarrow\mathbb{R}^n$, the (essential) supremum of $f$ is denoted by $\Vert f\Vert_{\infty} \Let \text{(ess)sup}\{\Vert f(k)\Vert,k\geq 0\}$. A function $\gamma:\mathbb\R_{0}^{+}\rightarrow\mathbb\R_{0}^{+}$, is said to be a class $\mathcal{K}$ function if it is continuous, strictly increasing, and $\gamma(0)=0$. A class $\mathcal{K}$ function $\gamma(\cdot)$ is said to be a class $\mathcal{K}_{\infty}$ if $\gamma(r)\rightarrow\infty$ as $r\rightarrow\infty$.
%The triplet $(\Omega, \sigalg, \PP)$ is a probability space
%\textcolor{red}{endowed with a filtration $\mathbb{F} = (\sigalg_s)_{s\geq 0}$ satisfying the usual conditions of completeness and right continuity \cite{}. Moreover, $(W_s)_{s \ge 0}$ and $(P_s)_{s\ge 0}$ are a $\widetilde{p}$-dimensional $\mathbb{F}$-Brownian motion and  a $\widetilde{q}$-dimensional $\mathbb{F}$-Poisson process, respectively, assumed te be independent of each other. The Poisson process $P_s\Let[P_s^1;\ldots;P_s^{\widetilde q}]$ model $\widetilde{q}$ kinds of events whose occurrences are also assumed to be independent of each other.}

\subsection{Discrete-Time Stochastic Control Systems}
We consider stochastic control systems in discrete time (dt-SCS) defined over a general state space adopted from \cite{hll1996} and characterized by the tuple
\begin{equation}
\label{eq:dt-SCS}
\Sigma =\left(X,W,U,\{U(x,\omega)|x\in X,\omega\in W\},Y,T_{\mathsf x},h\right),
\end{equation}
where $X$ is a Borel space as the state space of the system.
We denote by $(X, \mathcal B (X))$ the measurable space
with $\mathcal B (X)$  being  the Borel sigma-algebra on the state space. Sets
$W$ and $U$ are Borel spaces as the \emph{internal} and \emph{external} input spaces of the system.
The set $\{U(x,\omega)|x\in X,\omega\in W\}$ is a family of non-empty measurable subsets of $U$ with the property that
\begin{equation*}
K :=\{(x,\omega,\nu): x\in X,\omega \in W, \nu\in U(x)\}
\end{equation*}
is measurable in $X\times W\times U$. Intuitively, $U(x,\omega)$ is the set of inputs that are feasible at state $x\in X$ with the internal input $\omega\in W$. Set $Y$ is a Borel space as the output space of the system. Map $T_{\mathsf x}:\mathcal B(X)\times X\times W\times U\rightarrow[0,1]$,
is a conditional stochastic kernel that assigns to any $x \in X$, $\omega\in W$ and $\nu\in U(x,\omega)$ a probability measure $T_{\mathsf x}(\cdot | x,\omega,\nu)$
on the measurable space
$(X,\mathcal B(X))$
so that for any set $A \in \mathcal B(X), \PP_{x,\omega,\nu}(A) = \int_A T_{\mathsf x} (d\bar x|x,\omega,\nu)$, 
where $\PP_{x,\omega,\nu}$ denotes the conditional probability $\mathbb P(\cdot|x,\omega,\nu)$.
Finally,
$h:X\rightarrow Y$ is a measurable function that maps a state $x\in X$ to its output $y = h(x)\in Y$.

%\Sadegh{What kind of policies do we need here? from state or output?}

Given the dt-SCS in \eqref{eq:dt-SCS}, we are interested in \emph{Markov policies} to control the system.
\begin{definition}
A Markov policy for the dt-SCS $\Sigma$ in \eqref{eq:dt-SCS} is a sequence
$\rho = (\rho_0,\rho_1,\rho_2,\ldots)$ of universally measurable stochastic kernels $\rho_n$ \cite{BS96},
each defined on the input space $U$ given $X\times W$ and such that for all $(x_n,\omega_n)\in X\times W$, $\rho_n(U(x_n,\omega_n)|(x_n,\omega_n))=1$.
The class of all Markov policies is denoted by $\Pi_M$. 
\end{definition} 
%The class of discrete-time stochastic control systems studied in this paper is formalized
%in the following definition.
%
%\begin{definition}
%The employed control system here is the tuple $\Sigma=(\R^n,\R^m,\R^p,\mathcal{U},\mathcal{W},\R^q,h)$, in which 
%\begin{itemize}
%\item $\R^n$ is the state space;
%\item $\R^m$ is the external input space; 
%\item $\R^p$ is the internal input space; 
%\item $\mathcal{U}$ is a subset of the set of all \textcolor{red}{$\mathbb{F}$-progressively measurable processes} with values in $\R^m$;
%\item $\mathcal{W}$ is a subset of the set of all \textcolor{red}{$\mathbb{F}$-progressively measurable processes} with values in $\R^p$; 
%\item $\R^q$ is the output space;
%\item $h:\R^n\rightarrow\R^q$ is the output map. 
%\end{itemize}
%\end{definition}
%\Sadegh{The definition of the process should be modified. It is not precise...}

For given inputs $\omega(\cdot),\nu(\cdot)$, the stochastic kernel $T_{\mathsf x}$ captures the evolution of the state of the system.
This kernel features an equivalent dynamical representation: there exists a measurable function $f_{a}:X\times W\times U\times V_{\varsigma} \rightarrow X$ such that the evolution of the state of the system can be written as
\begin{equation*}
x(k+1) = f_a(x(k),\omega(k),\nu(k),\varsigma(k)),
\end{equation*}
where $\{\varsigma(k):\Omega\rightarrow V_{\varsigma},\,\,k\in\N\}$ is a sequence of independent and identically distributed (i.i.d.) random variables on the set $V_\varsigma$.
In this paper we assume that the state space $X$ is a subset of $\mathbb R^n$ and are interested in the specific form of the function
$$f_a(x,\omega,\nu,\varsigma) = f(x,\omega,\nu)+g(x)\varsigma. $$
Therefore, the dt-SCS $\Sigma$ in \eqref{eq:dt-SCS} can be described as:
%A discrete-time stochastic control system $\Sigma$ satisfies
\begin{align}\label{Eq1a}
\Sigma:\left\{\hspace{-1.5mm}\begin{array}{l}x(k+1)=f(x(k),\omega(k),\nu(k))+g(x(k))\varsigma(k),\\
y(k)=h(x(k)),\end{array}\right.
\end{align}
for any $x(k)\in X, \omega(k)\in W,$ and $\nu(k)\in U(x(k),\omega(k))$.
Note that $T_{\mathsf x}$ in \eqref{eq:dt-SCS} contains the information of functions $f$ and $g$ and the distribution of noise $\varsigma(\cdot)$ in the dynamical representation \eqref{Eq1a}.

%\Sadegh{Making $g$ also a function of u and/or w? Does it change the results that much?}

For the sake of simplicity, we also assume that the set of valid inputs is the whole input space: $U(x,\omega) = U$ for all $x\in X$ and $\omega\in W$, but the obtained results are generally applicable.
We associate respectively to $U$ and $W$ the sets $\mathcal U$ and $\mathcal W$ to be collections of sequences $\{\nu(k):\Omega\rightarrow U,\,\,k\in\N\}$ and $\{\omega(k):\Omega\rightarrow W,\,\,k\in\N\}$, in which $\nu(k)$ and $\omega(k)$ are independent of $\varsigma(t)$ for any $k,t\in\mathbb N$ and $t\ge k$.

%$\PP$-almost surely ($\PP$-a.s.) for any $\nu\in\mathcal{U}$ and any $\omega\in\mathcal{W}$, in which stochastic processes $\xi:\Omega \times \R_{\geq0} \rightarrow \R^n$ and $\zeta:\Omega \times \R_{\geq0} \rightarrow \R^q$ are called a \textit{solution process} and \textit{output trajectory} of $\Sigma$, respectively. We here call the tuple $(\xi,\zeta,\nu,\omega)$ a \textit{trajectory}
%of $\Sigma$, that satisfies \eqref{Eq1a} $\PP$-a.s.. We also write $\xi(a,\nu,\omega)$ to denote the value of the solution process under the input trajectories $\nu$ and $\omega$ from initial condition $\xi(a,\nu,\omega)=a$ $\PP$-a.s., in which $a$ is a random variable that is $\sigalg_0$-measurable. Furthermore, we denote by $\zeta(a,\nu,\omega)$ the output trajectory corresponding to the solution process $\xi(a,\nu,\omega)$.

For any initial state $a\in X$, $\nu(\cdot)\in\mathcal{U}$, and $\omega(\cdot)\in\mathcal{W}$,
the random sequences $x_{a\omega\nu}:\Omega \times\N \rightarrow X$  and $y_{a\omega\nu}:\Omega \times \N \rightarrow Y$ that satisfy \eqref{Eq1a}
are called respectively the \textit{solution process} and \textit{output trajectory} of $\Sigma$ under external input $\nu$, internal input $\omega$ and initial state $a$.
We here call the tuple $(\omega,\nu,x_{a\omega\nu},y_{a\omega\nu})$ a \textit{trajectory}
of $\Sigma$.
%that satisfies \eqref{Eq1a} $\PP$-a.s.. We also write $\xi(a,\nu,\omega)$ to denote the value of the solution process under the input trajectories $\nu$ and $\omega$ from initial condition $\xi(a,\nu,\omega)=a$ $\PP$-a.s., in which $a$ is a random variable that is $\sigalg_0$-measurable. Furthermore, we denote by $\zeta(a,\nu,\omega)$ the output trajectory corresponding to the solution process $\xi(a,\nu,\omega)$.

\section{Stochastic Pseudo-Simulation and Simulation Functions}
In this section we first introduce a notion of so-called pseudo-simulation functions for the discrete-time stochastic control systems with both internal and external inputs and then define the stochastic simulation functions for systems with only external input. These definitions can be used to quantify closeness of two dt-SCS with the same internal input and output spaces.

%\begin{definition}\label{Def1}
%Consider two discrete-time stochastic control systems
%$\Sigma=(\R^n,\R^m,\R^p,\mathcal{U},\mathcal{W},\R^q,h)$ and
%$\widehat\Sigma=(\R^{\hat n},\R^{\hat m},\R^p,\hat{\mathcal{U}},\mathcal{W},\R^q,\hat h)$
%with the same internal input and output space
%dimension. A function $V:\R^{n}\times \R^{\hat n}\to\R_{\ge0}$ is
%called a stochastic pseudo-simulation function from $\Sigma$ to $\widehat\Sigma$ (SPSF), if
%$\forall x\in\R^n$, $\forall \hat x\in\R^{\hat n}$, $\exists\alpha\in \mathcal{K}_{\infty}$, such that
%\begin{eqnarray}
%\alpha(\Vert h(x)-\hat h(\hat x)\Vert)\le V(x,\hat x),	
%\end{eqnarray}
%and $\forall \hat u\in\R^{\hat m}$, $\forall\hat w\in\R^p$, $\exists u\in\R^m$, $\forall w\in\R^p$, that
%\begin{align}\notag
%&\EE \Big[V(x(k+1),\hat{x}(k+1))\,|\,[x(k);\hat{x}(k); w(k); \hat{w}(k);\hat{u}(k)]\Big]\\\notag
%&-V(x(k),\hat{x}(k))\leq-\kappa\Big(V(x(k),\hat{x}(k))\Big)+ \rho_{\mathrm{int}}(\Vert w-\hat w\Vert)\\
%&+\rho_{\mathrm{ext}}(\Vert\hat u\Vert)+\psi,
%\end{align}
%for some $\kappa\in \mathcal{K}$, $\rho_{\mathrm{int}},\rho_{\mathrm{ext}} \in \mathcal{K}_{\infty}$, and $\psi \in\textcolor{red}{\R^+}$.
%\end{definition}

\begin{definition}\label{Def1}
Consider dt-SCS
$\Sigma =(X,W,U,Y,T_{\mathsf x},h)$ and
$\widehat\Sigma =(\hat X,W,\hat U,Y,\hat T_{\mathsf x},\hat h)$
with the same internal input and output spaces.
A function $V:X\times\hat X\to\R_{\ge0}$ is
called a stochastic pseudo-simulation function (SPSF) from  $\widehat\Sigma$ to $\Sigma$ if
\begin{itemize}
\item  $\exists \alpha\in \mathcal{K}_{\infty}$ such that
\begin{eqnarray}
\forall x\in X,\forall \hat x\in\hat X,\quad \alpha(\Vert h(x)-\hat h(\hat x)\Vert)\le V(x,\hat x),\label{eq:V_dec1}
\end{eqnarray}
\item $\forall x\in X,\hat x\in\hat X,\hat\nu\in\hat U$, and $\forall \hat{\omega}\in\hat W$, $\exists \nu\in U$ such that $\forall \omega\in W$
\begin{align}\notag
&\EE \Big[V(x(k+1),\hat{x}(k+1))\,\big|\,x(k),\hat{x}(k), \omega(k)=\omega,\\\notag
\,\,\,&,\hat{\omega}(k)=\hat{\omega},\nu(k)=\nu,\hat{\nu}(k)=\hat{\nu}\Big]-V(x(k),\hat{x}(k))\leq\\
&\!-\!\kappa(V(x(k),\hat{x}(k)))\!+\! \rho_{\mathrm{int}}(\Vert\omega-\hat\omega\Vert)\!+\!\rho_{\mathrm{ext}}(\Vert\hat\nu\Vert)\!+\!\psi,
\label{eq:V_dec}
\end{align}
for some $\kappa\in \mathcal{K}$, $\rho_{\mathrm{int}},\rho_{\mathrm{ext}} \in \mathcal{K}_{\infty}\cup \{0\}$, and $\psi \in\R_{\ge 0}$.
\end{itemize}
\end{definition}

We utilize notation $\widehat\Sigma\preceq_{\mathcal{PS}}\Sigma$ if there exists a pseudo-simulation function $V$ from $\widehat\Sigma$ to $\Sigma$, in which control system  $\widehat\Sigma$ is  considered as an abstraction of concrete (original) system $\Sigma$.
The second condition above implies implicitly existence of a function $\nu=\nu_{\hat \nu}(x,\hat x,\hat \nu,\hat \omega)$ for satisfaction of \eqref{eq:V_dec}. This function is called the \emph{interface function} and can be used to refine a synthesized policy $\hat\nu$ for $\widehat\Sigma$ to a policy $\nu$ for $\Sigma$.

In this paper we study interconnected discrete-time stochastic control systems without internal inputs, resulting from the interconnection of discrete-time stochastic control subsystems having both internal and external signals. In this case, the interconnected dt-SCS reduces to the tuple $(X,U,Y,T_{\mathsf x},h)$.  Thus we modify the above notion for systems without internal inputs.

%\Sadegh{Define interface function for the above SPSF?}

%\Sadegh{Is this definition only for eliminating the internal input?}
\begin{definition}\label{Def2}
Consider two dt-SCS
$\Sigma =(X,U,Y,T_{\mathsf x},h)$ and
$\widehat\Sigma =(\hat X,\hat U,Y,\hat T_{\mathsf x},\hat h)$
with the same output spaces.
A function $V:X\times\hat X\to\R_{\ge0}$ is
called a \emph{stochastic simulation function} (SSF) from $\widehat\Sigma$  to $\Sigma$ if
\begin{itemize}
\item $\exists \alpha\in \mathcal{K}_{\infty}$ such that
\begin{eqnarray}
\label{eq:lowerbound2}
\forall x\in X,\forall \hat x\in\hat X,\quad \alpha(\Vert h(x)-\hat h(\hat x)\Vert)\le V(x,\hat x),
\end{eqnarray}
\item $\forall x\in X,\hat x\in\hat X,\hat\nu\in\hat U$, $\exists \nu\in U$ such that
\begin{align}
&\EE \Big[V(x(k\!+\!1),\hat{x}(k\!+\!1))\big|x(k),\hat{x}(k),\nu(k)\!\!=\!\!\nu, \hat{\nu}(k)\!\!=\!\!\hat{\nu}\Big]\nonumber\\
&\!-\!V(x(k),\hat{x}(k))\!\leq\!\!-\kappa(V(x(k),\hat{x}(k)))
\!+\!\rho_{\mathrm{ext}}(\Vert\hat\nu\Vert)\!+\!\psi,\label{eq:martingale2}
\end{align}
for some $\kappa\in \mathcal{K}$, $\rho_{\mathrm{ext}} \in \mathcal{K}_{\infty}\cup \{0\}$, and $\psi \in\R_{\ge 0}$.
\end{itemize}
\end{definition}

%\begin{definition}\label{Def2}
% Let $\Sigma=(\R^n,\R^m,\mathcal{U},\R^q,h)$ and $\widehat\Sigma=(\R^{\hat n},\R^{\hat m},\hat{\mathcal{U}},\R^q,\hat h)$ be two discrete-time stochastic control systems such that $h$ and $\hat h$ are of equal dimension. A function $V:\R^{n}\times \R^{\hat n}\to\R_{\ge0}$ is
%called a stochastic simulation function from $\widehat\Sigma$ to $\Sigma$ (SSF), if
%\begin{eqnarray}
%\alpha(\Vert h(x)-\hat h(\hat x)\Vert)\le V(x,\hat x), \label {Con1}	
%\end{eqnarray}
%for any $x\in\R^n$, $\hat x\in\R^{\hat n}$,  and some $\alpha\in \mathcal{K}_{\infty}$, and
%\begin{align}\notag\label {Con2}
%&\EE\Big[V(x(k+1),\hat{x}(k+1))\,|\,[x(k);\hat{x}(k); w(k); \hat{w}(k);\hat{u}(k)]\Big]\\ &-V(x(k),\hat{x}(k))\leq-\kappa\Big(V(x(k),\hat{x}(k))\Big)+\rho_{\mathrm{ext}}(\Vert\hat u\Vert)+\psi,
%\end{align}
%$\forall \hat u\in\R^{\hat m}$, $\forall\hat w\in\R^p$, $\exists u\in\R^m$, $\forall w\in\R^p$.
%Stochastic control system $\widehat\Sigma$ is simulated by $\Sigma$, or $\Sigma$ simulates $\widehat\Sigma$, denoted by $\Sigma\preceq_{\mathcal{S}}\widehat\Sigma$, if there exists a stochastic simulation function from $\Sigma$ to $\widehat\Sigma$. In the next theorem, the importance of the existence
%of a simulation function by quantifying the error between the output behaviors of  $\Sigma$ and its abstraction  $\widehat\Sigma$ is shown.
%\end{definition}

%\Sadegh{The input and disturbance are not included in the next theorem.}
The next theorem shows usefulness of SSF in comparing output trajectories of two dt-SCS in a probabilistic sense.  
\begin{theorem}\label{Thm1}
Let $\Sigma$ and $\widehat\Sigma$ be two dt-SCS with the same output spaces.
%Let $\Sigma$ and $\widehat\Sigma$ be two discrete-time stochastic control systems with the output space dimension,
Suppose $V$ is an SSF from $\widehat\Sigma$ to $\Sigma$, and there exists a constant $0<\widehat\kappa<1$ such that the function $\kappa \in \mathcal{K}$ in \eqref{eq:martingale2} satisfies $\kappa(r)\geq\widehat\kappa r$ $\forall r\in\R_{\geq0}$. For any external input trajectory $\hat\nu(\cdot)\in\mathcal{\hat U}$ that preserves Markov property for the closed-loop $\widehat\Sigma$, and for any random variables $a$ and $\hat a$ as the initial states of the two dt-SCS,
there exists an input trajectory $\nu(\cdot)\in\mathcal{U}$ of $\Sigma$ through the interface function associated with $V$ such that the following inequality holds provided that there exists a constant $\widehat\psi\geq0$ satisfying  $\widehat\psi\geq \rho_{\mathrm{ext}}(\Vert\hat \nu\Vert_{\infty})+\psi$:

\begin{align}\label{Eq2a}
&\PP\left\{\sup_{0\leq k\leq T}\Vert y_{a\nu}(k)-\hat y_{\hat a \hat\nu}(k)\Vert\geq\varepsilon\,|\,[a;\hat a]\right\}\\\notag
&\leq
\begin{cases}
1-\Big(1-\frac{V(a,\hat a)}{\alpha\left(\varepsilon\right)}\Big)\Big(1-\frac{\widehat\psi}{\alpha\left(\varepsilon\right)}\Big)^{T} & \!\!\text{if}~\alpha\left(\varepsilon\right)\!\geq\!\frac{\widehat\psi}{\widehat\kappa},\\
\Big(\frac{V(a,\hat a)}{\alpha\left(\varepsilon\right)}\Big)\!(1\!-\!\widehat\kappa)^T\!\!+\!\Big(\frac{\widehat\psi}{\widehat\kappa\alpha\left(\varepsilon\right)}\Big)\!(1\!-\!(1\!-\!\widehat\kappa)^T) & \!\!\text{if}~\alpha\left(\varepsilon\right)\!<\!\frac{\widehat\psi}{\widehat\kappa}.
\end{cases}
\end{align}
\end{theorem}
\begin{proof}
Since $V$ is an SSF from $\widehat\Sigma$ to $\Sigma$, we have
\begin{align}
\PP&\left\{\sup_{0\leq k\leq T}\Vert y_{a\nu}(k)-\hat y_{\hat a \hat\nu}(k)\Vert\geq\varepsilon\,|\,[a;\hat a]\right\}\nonumber\\
=\PP&\left\{\sup_{0\leq k\leq T}\alpha\left(\Vert y_{a\nu}(k)-\hat y_{\hat a \hat\nu}(k)\Vert\right)\geq\alpha(\varepsilon)\,|\,[a;\hat a]\right\}\nonumber\\
\leq\PP&\left\{\sup_{0\leq k\leq T}V\left(x_{a\nu}(k),\hat x_{\hat a \hat\nu}(k)\right)\geq\alpha(\varepsilon)\,|\,[a;\hat a]\right\}.\label{eq:supermart}
\end{align}
The equality holds due to $\alpha$ being a $\mathcal K_\infty$ function. The inequality  is also true due to condition \eqref{eq:lowerbound2} on the SSF $V$. The results follows by applying Theorem 3 in \cite[pp. 81]{1967stochastic} to \eqref{eq:supermart} and utilizing inequality \eqref{eq:martingale2}.
\end{proof}

The results shown in Theorem \ref{Thm1} provide closeness of output behaviours of two systems in finite-time horizon. We can extend the result to infinite-time horizon provided that constant $\hat{\psi}=0$ as the following. 
%The next corollary presents a similar result as Theorem \ref{Thm1} for the discrete-time stochastic control systems, but by considering an infinite time horizon.
\begin{corollary}
Let $\Sigma$ and $\widehat\Sigma$ be two dt-SCS with the same output spaces.
Suppose $V$ is an SSF from $\widehat\Sigma$ to $\Sigma$ such that $\rho_{\mathrm{ext}}(\cdot)\equiv0$ and $\psi = 0$.
For any external input trajectory $\hat\nu(\cdot)\in\mathcal{\hat U}$ preserving Markov property for the closed-loop $\widehat \Sigma$, and for any random variables $a$ and $\hat{a}$ as the initial states of the two dt-SCS, there exists $\nu(\cdot)\in{\mathcal{U}}$ of $\Sigma$ through the interface function associated with $V$ such that the following inequality holds:
\begin{align}\nonumber
\PP\left\{\sup_{0\leq k< \infty}\Vert y_{a\nu}(k)-\hat y_{\hat a 0}(k)\Vert\ge \varepsilon\,|\,[a;\hat a]\right\}\leq\frac{V(a,\hat a)}{\alpha\left(\varepsilon\right)}.
\end{align}
\end{corollary}
\begin{proof}
Since $V$ is an SSF from $\widehat\Sigma$ to $\Sigma$ with $\rho_{\mathrm{ext}}(\cdot)\equiv0$ and $\psi = 0$, for any $x(k)\in X$ and $\hat x(k)\in\hat X$ and any $\hat\nu(k)\in\hat U$, there exists $\nu(k)\in U$ such that
\begin{align}\notag
&\EE \Big[V(x(k+1),\hat x(k+1))\,|\,x(k),\hat x(k),\nu(k), \hat \nu(k)\Big] \\
&-V((x(k),\hat x(k))\leq-\kappa (V(x(k),\hat x(k)),
\end{align}
showing that $V\left(x_{a\nu}(k),\hat x_{\hat a \hat\nu}(k)\right)$ is a nonnegative supermartingale \cite{oksendal2013stochastic}.
Following the same reasoning as in the proof of Theorem  \ref{Thm1} we have
\begin{align}\nonumber
\PP&\left\{\sup_{0\leq k< \infty}\Vert y_{a\nu}(k)-\hat y_{\hat a \hat \nu}(k)\Vert\ge \varepsilon\,|\,[a;\hat a]\right\}\\\notag
=\PP&\left\{\sup_{0\leq k< \infty}\alpha\Big(\Vert y_{a\nu}(k)-\hat y_{\hat a \hat \nu}(k)\Vert\Big)\ge \alpha(\varepsilon)\,|\,[a;\hat a]\right\}\\\notag
\leq\PP&\left\{\sup_{0\leq k< \infty}V(x_{a\nu}(k),\hat x_{\hat a\hat \nu}(k))\ge \alpha(\varepsilon)\,|\,[a;\hat a]\right\}\!\leq\!\!\frac{V(a,\hat a)}{\alpha(\varepsilon)},
\end{align}
where the last inequality is due to the nonnegative supermartingale property \cite{1967stochastic}.
\end{proof}

%Note that the notion of simulation probabilistic closeness in this Theorem is based on the trajectory rather than probability distribution. However, for the safety and reachability, the next result shows that our notion of closeness actually implies closeness in distribution.

The stochastic simulation function defined before can be used to guarantee an upper bound on the probability of the maximum difference in output trajectories.
%In general, we would like to have a tight upper bound, which corresponds to small stochastic simulation function.
This idea can be used in conjunction with stochastic safety/reachability analysis of the systems, which is discussed next.

Suppose that $V$ is a stochastic simulation function from $\widehat{\Sigma}$ to $\Sigma$.
Then for any input strategy $\hat{\nu}$ of the system $\widehat\Sigma$ there exists an input strategy $\nu$ of $\Sigma$, such that the following probability is bounded
\begin{equation*}
\PP\left\{\sup_{0\leq k\leq T}\Vert y_{a\nu}(k)-\hat y_{\hat a \hat\nu}(k)\Vert\geq\varepsilon\,|\,[a;\hat a]\right\}\le \delta,
\end{equation*}
with $\delta$ being defined in Theorem \ref{Thm1} based on $\varepsilon$ and $T$.
Given the unsafe set $A_1$ for $\Sigma$, we can construct another set $A_2$, which is the $\varepsilon$ neighborhood of $A_1$, i.e.,
\begin{equation*}
A_2 = \{y'|\exists y\in A_1, \|y'-y\|\le\varepsilon \}.
\end{equation*} 
Now, we can provide the following corollary.

\begin{corollary}
Suppose $V$ is an SSF from $\widehat{\Sigma}$ to $\Sigma$.
For any input $\hat{\nu}(\cdot)$ there exists $\nu(\cdot)$ such that the following inequality holds:
\begin{equation*}
\PP\{\exists k\le T,\,\,y_{a\nu}(k)\in A_1\}\le \PP\{\exists k\le T,\,\,\hat y_{\hat a \hat \nu}(k)\in A_2\}+\delta.
\end{equation*}
\end{corollary}

\begin{proof}
Denote the events $\mathcal E_1 := \{\exists k\le T,\,\,y_{a\nu}(k)\in A_1\}$ and
$\mathcal E_2 := \{\exists k\le T,\,\,\hat y_{\hat a \hat \nu}(k)\in A_2\}$.
Then we have
\begin{equation*}
\PP\{\mathcal E_1\} = \PP\{\mathcal E_1\cap\mathcal E_2\}+\PP\{\mathcal E_1\cap\bar{\mathcal E}_2\}
\le \PP\{\mathcal E_2\}+\PP\{\mathcal E_1\cap\bar{\mathcal E}_2\},
\end{equation*}
where $\bar{\mathcal E}_2$ is the complement of $\mathcal E_2$.
Notice that the term $\PP\{\mathcal E_1\cap\bar{\mathcal E}_2\}$ is bounded by $\delta$ due to the above results, which concludes the proof.
%\begin{equation*}
%\PP\{\mathcal E_2\} = \PP\{\mathcal E_2\cap\mathcal E_1\}+\PP\{\mathcal E_2\cap\bar{\mathcal E}_1\}
%\le \PP\{\mathcal E_1\}+\PP\{\mathcal E_2\cap\bar{\mathcal E}_1\},
%\end{equation*}
\end{proof}

\section{Compositional Abstractions for Interconnected Systems}
Here, we first provide a formal definition of interconnection between discrete-time stochastic control systems.

\subsection{Interconnected Stochastic Control Systems}
Consider a complex stochastic control system  $\Sigma$ composed of $N\in\mathbb N_{\geq1}$ stochastic control subsystems $\Sigma_i$ interconnected with each other as follows:  
\begin{eqnarray}\nonumber
\Sigma_i=(X_i,W_i,U_i,Y_i,T_{\mathsf x_i},h_i),
&i\in [1;N],
\end{eqnarray}
with partitioned internal inputs and outputs
\begin{align}\notag
\omega_i&=\intcc{\omega_{i1};\ldots;\omega_{i(i-1)};\omega_{i(i+1)};\ldots;\omega_{iN}},\\\label{config1}
y_i&=\intcc{y_{i1};\ldots;y_{iN}},
\end{align}
and also output space and function 
\begin{align}\notag
h_i(x_i)&=\intcc{h_{i1}(x_i);\ldots;h_{iN}(x_i)},\\\label{config2}
Y_i&=\prod_{j=1}^{N}Y_{ij}.
\end{align}
We interpret the outputs $y_{ii}$ as \emph{external} ones, whereas the outputs $y_{ij}$ with $i\neq j$ are \emph{internal} ones which are used to define the interconnected stochastic control systems. In particular, we assume that the dimension of $\omega_{ij}$ is equal to the dimension of $y_{ji}$. If there is no connection from stochastic control subsystem $\Sigma_{i}$ to $\Sigma_j$, then we assume that the connecting output function is identically zero for all arguments, i.e., $h_{ij}\equiv 0$. Now, we define the \emph{interconnected stochastic control systems} as the following.

\begin{definition}
Consider $N\in\N_{\geq1}$ stochastic control subsystems $\Sigma_i=(X_i,W_i,U_i,Y_i,T_{\mathsf x_i},h_i)$, $i\in[1;N]$, with the input-output configuration as in \eqref{config1} and \eqref{config2}. The
interconnection of  $\Sigma_i$ for any $i\in [1,\ldots,N]$, is the interconnected stochastic control system $\Sigma=(X,U,Y,T_{\mathsf x},h)$, denoted by
$\mathcal{I}(\Sigma_1,\ldots,\Sigma_N)$, such that $X:=\prod_{i=1}^{N}X_i$,  $U:=\prod_{i=1}^{N}U_i$, function $f_a:=\prod_{i=1}^{N}f_{ai}$, characterizing the stochastic kernel $T_{\mathsf x}$ based on those of subsystems (i.e. $f_{ai}$), $Y:=\prod_{i=1}^{N}Y_{ii}$, and $h=\prod_{i=1}^{N}h_{ii}$, subjected to the following constraint:
\begin{align}
\omega_{ij}=y_{ji} \hspace{0.2cm}\forall i,j\in[1;N],i\neq j.
\end{align}
\end{definition}
\subsection{Compositional Abstractions of Interconnected Systems}
This subsection contains one of the main contributions of the paper. We assume that we are given $N$ stochastic control subsystems $$\Sigma_i=(X_i,W_i,U_i,Y_i,T_{\mathsf x_i},h_i),$$ together with their corresponding abstractions $\widehat\Sigma_i=(\hat X_i,W_i,\hat U_i, Y_i,\hat T_{\mathsf x_i},\hat h_i)$  with SPSF $V_i$ from $\widehat\Sigma_i$ to $\Sigma_i$. 
For providing the main compositionality result of the paper, we raise the following assumption.

\begin{assumption}\label{Asum1}
For any $i,j\in[1;N]$, $i\neq j$, there exist $\mathcal{K}_\infty$ functions $\gamma_i$ and constants $\lambda_i\in\R_{>0}$ and $\delta_{ij}\in\R_{\geq0}$ such that for any $s\in\R_{\geq0}$
\begin{align}\label{assumption}
&\kappa_{i}(s) \geq\lambda_i\gamma_i(s)\\
&h_{ji}\equiv 0\implies \delta_{ij}=0 \text{ and }\\
&h_{ji}\not\equiv 0\implies \rho_{i\mathrm{int}}((N-1)\alpha_j^{-1}(s))\leq\delta_{ij}\gamma_j(s),\label{Con1a}
\end{align}
where $\kappa_i$, $\alpha_j$, and $\rho_{i\mathrm{int}}$ represent the corresponding  $\mathcal{K}$ and $\mathcal{K}_\infty$ functions of $V_i$ appearing in Definition \ref{Def1}. Prior to presenting the next theorem, we define 
$\Lambda\Let\text{diag}(\lambda_1,\ldots,\lambda_N)$, $\Delta\Let\{\delta_{ij}\}$, where $\delta_{ii}=0$ $\forall i\in[1;N]$, and $\Gamma(\as)\Let[\gamma_1(s_1);\ldots;\gamma_N(s_N)]$, where $\as=[s_1;\ldots;s_N]$.
In the next theorem, we leverage a small-gain type condition to quantify the error between the interconnection of stochastic control subsystems and that of their abstractions in a compositional way.
\end{assumption}
\begin{theorem}\label{Thm5}
Consider the interconnected stochastic control system
$\Sigma=\mathcal{I}(\Sigma_1,\ldots,\Sigma_N)$ induced by $N\in\N_{\geq1}$ stochastic
control subsystems~$\Sigma_i$. Suppose that each stochastic control subsystem $\Sigma_i$ admits an abstraction $\widehat \Sigma_i$ with the corresponding SPSF $V_i$. If Assumption~\ref{Asum1} holds and there exists a vector $\mu\in\R^N_{>0}$ such that the inequality
\begin{align}
\mu^T(-\Lambda+\Delta)< 0
\end{align}
is also met, then 
\begin{align}\notag
V(x,\hat x)\Let\sum_{i=1}^N\mu_iV_i(x_i,\hat x_i)
\end{align}
is an SSF function from $\widehat \Sigma=\mathcal{I}(\widehat \Sigma_1,\ldots,\widehat\Sigma_N)$ to $\Sigma=\mathcal{I}(\Sigma_1,\ldots,\Sigma_N)$. 	
\end{theorem}

The proof is similar to that of Theorem 4.5 in \cite{zamani2014compositional}, and is omitted here due to lack of space.

\section{Discrete-Time Linear Stochastic Control Systems}
In this section, we focus on a class of discrete-time linear stochastic control systems, defined as follows:
\begin{align}\label{linear}
\Sigma:\left\{\hspace{-1.5mm}\begin{array}{l}x(k+1)=Ax(k)+B\nu(k)+D\omega(k)+F\varsigma(k),\\
y(k)=Cx(k),\end{array}\right.
\end{align}
%where
%\begin{eqnarray}\notag
%A\in\R^{n\times n},
%B\in\R^{n\times m},
%D\in\R^{n\times p},
%F\in\R^{n\times n},
%C\in\R^{q\times n}.
%\end{eqnarray}
where the additive noise $\varsigma(k)$ is a sequence of independent random vectors with multivariate standard normal distributions.
%where $\varsigma(k)$ is a $\widetilde{q}$-dimensional white-noise sequence with
%independent random variables satisfying $ \EE(\varsigma(k))=0$ and $\text{Var}(\varsigma(k))=I_{\widetilde{q}}$. 
We use the tuple $\Sigma=\left(A,B,C,D,F\right)$ to refer to the class of systems in \eqref{linear}. Here, we provide conditions under which a candidate $V$ is an SPSF function facilitating the construction of an abstraction $\hat \Sigma$. 

Let us assume that there exist matrix
$K$ and positive definite matrix $M$ such that the matrix inequalities
\begin{align}\label{Eq8a}
&C^TC\preceq M, \\
&\Big((1+\pi)(A+BK)^TM(A+BK)-M\Big)\preceq -\widehat\kappa M, \label{Eq9a}
\end{align}
hold for some positive constants $\pi$ and $0<\widehat\kappa<1$. We employ the following quadratic SPSF
\begin{align}\
V(x,\hat x)=(x-P\hat x)^TM(x-P\hat x), \label{Eq10a}
\end{align}
where $P\in\R^{n\times\hat n}$ is a matrix of appropriate dimension. 
Assume that the equalities 
\begin{eqnarray}\label{Eq11a}
AP&=&P\hat A-BQ\\
D&=&P\hat D-BS\label{Eq12a}\\
%B\widetilde{R}&=&P\hat B\label{Eq13a}\\
CP&=&\hat C,\label{Eq14a}
%F_iP&=&P \hat F_i,~\forall i\in[1;\widetilde{q}], \label{Eq15a}
\end{eqnarray}
hold for some matrices $Q$ and $S$ of appropriate dimensions and possibly with the lowest possible $\hat n$. In the next theorem, we show that under the aforementioned conditions $V$ in (\ref{Eq10a}) is an SPSF from $\widehat \Sigma$ to $\Sigma$.

\begin{theorem}\label{Thm3}
Let $\Sigma=(A,B,C,D,F)$ and $\widehat \Sigma=(\hat A,\hat B,\hat C,\hat
D,\hat F)$ be two discrete-time linear stochastic control subsystems with two independent additive noises. 
%satisfying $ \EE(\varsigma(k))=\EE(\hat\varsigma(k))=0$ and $\text{Var}(\varsigma(k))=\text{Var}(\hat\varsigma(k))=I_{\widetilde{q}}$. 
Suppose that there exist matrices $M$, $K$, $P$, $Q$, and $S$ satisfying (\ref{Eq8a}),~(\ref{Eq9a}), (\ref{Eq11a}), (\ref{Eq12a}), and (\ref{Eq14a}). Then, $V$ defined in (\ref{Eq10a}) is an SPSF from $\widehat\Sigma$ to $\Sigma$.
\end{theorem}

\begin{proof}
Here, we show that $\forall x$, $\forall \hat x$, $\forall \hat \nu$, $\forall \hat \omega$, $\exists\nu$, $\forall \omega$, such that $V$ satisfies $\Vert Cx-\hat C\hat x\Vert^2\le V(x,\hat x)$ and
\begin{align}\notag\label{Eq16a}
&\EE \Big[V(x(k+1),\hat x(k+1))\,|\,x(k),\hat x(k), \omega(k)\!=\!\omega, \hat \omega(k)\!=\!\hat \omega,\\\notag
&, \hat \nu(k)\!=\! \hat \nu \Big]-V(x(k),\hat x(k))\\\notag 
&\leq-\widehat\kappa (V(x(k),\hat x(k)))+(1+\frac{2}{\pi}+\frac{\pi}{2}){\Vert\sqrt{M}D\Vert^2}\Vert \omega-\hat \omega\Vert^2\\\notag
&+(1+\frac{2}{\pi}+\frac{2}{\pi})\Vert\sqrt{M}(B\widetilde R-P\hat B)\Vert^2\Vert\hat \nu\Vert^2\\
&+\text{Tr}\Big(F^TMF+\hat F^TP^TMP\hat F\Big).
\end{align}
According to (\ref{Eq14a}), we have $\Vert Cx-\hat C\hat x\Vert^2=(x-P\hat
x)^TC^TC(x-P\hat x)$. By applying (\ref{Eq8a}), it can be easily verified that 
$\Vert Cx-\hat C\hat x\Vert^2\le V(x,\hat x)$ holds $\forall x$, $\forall \hat x$. Now, we show inequality (\ref{Eq16a}). Given any $x$, $\hat x$, $\hat \nu$, and $\hat \omega$, we choose $\nu$ via the following \emph{linear interface function}:
\begin{align}\label{Eq17aa}
\nu=\nu_{\hat \nu}(x,\hat x,\hat \nu,\hat w):=K(x-P\hat x)+Q\hat x+\widetilde R\hat \nu+S \hat \omega,
\end{align}
for some matrix $\widetilde R$ of appropriate dimension. By Employing equations (\ref{Eq11a}), (\ref{Eq12a}), and the definition of the interface function in \eqref{Eq17aa}, we simplify
\begin{align}\notag
Ax&+B\nu_{\hat \nu}(x,\hat x, \hat \nu,\hat \omega)+D\omega -P(\hat A\hat x+\hat B\hat \nu+\hat D\hat \omega)\\\notag &+\Big(F\varsigma(k)-P\hat F\hat\varsigma(k)\Big)
\end{align}
to $(A+BK)(x-P\hat x)+D(\omega-\hat \omega)+(B\widetilde R-P\hat B)\hat \nu+\Big(F\varsigma(k)-P\hat F\hat\varsigma(k)\Big)$. 
%Using $ \EE(\varsigma(k))=\EE(\hat\varsigma(k))=0$ and $\text{Var}(\varsigma(k))=\text{Var}(\hat\varsigma(k))=I_{\widetilde{q}}$, 
One obtains:
\begin{align}\notag
&\EE \Big[V(x(k+1),\hat x(k+1))\,|\,x(k),\hat x(k), \omega(k)\!=\!\omega, \hat \omega(k)\!=\!\hat \omega,\\\notag
&, \hat \nu(k)\!=\! \hat \nu \Big]-V(x(k),\hat x(k))\\\notag 
&=(x-P\hat x)^T\Big[(A+BK)^TM(A+BK)-M\Big](x-P\hat x)\\\notag
&+\Big[2(x-P\hat x)^T(A+BK)^T\Big]M\Big[D(\omega-\hat \omega)\Big]\\\notag
&+\Big[2(x-P\hat x)^T(A+BK)^T\Big]M\Big[(B\widetilde R-P\hat B)\hat \nu\Big]\\\notag
&+\Big[2(\omega-\hat \omega)^TD^T\Big]M\Big[(B\widetilde R-P\hat B)\hat \nu\Big]+{\Vert\sqrt{M}D(\omega-\hat \omega)}\Vert^2\\\notag
&+{\Vert\sqrt{M}(B\widetilde R-P\hat B)\hat \nu}\Vert^2+\text{Tr}\big(F^TMF+\hat F^TP^TMP\hat F\big).
\end{align}
Using Young's inequality \cite{young1912classes} as $ab\leq \frac{\pi}{2}a^2+\frac{1}{2\pi}b^2,$ for any $a,b\geq0$ and any $\pi>0$, and by employing Cauchy-Schwarz inequality and \eqref{Eq9a}, one obtains the following upper bound:
\begin{align}\notag
&\EE \Big[V(x(k+1),\hat x(k+1))\,|\,x(k),\hat x(k), \omega(k)\!=\!\omega, \hat \omega(k)\!=\!\hat \omega,\\\notag
&, \hat \nu(k)\!=\! \hat \nu \Big]-V(x(k),\hat x(k))\\\notag 
&\leq-\widehat\kappa (V(x,\hat x))+(1+\frac{2}{\pi}+\frac{\pi}{2}){\Vert\sqrt{M}D\Vert^2}\Vert \omega-\hat \omega\Vert^2\\\notag
&+(1+\frac{2}{\pi}+\frac{2}{\pi})\Vert\sqrt{M}(B\widetilde R-P\hat B)\Vert^2\Vert\hat \nu\Vert^2\\
&+\text{Tr}\Big(F^TMF+\hat F^TP^TMP\hat F\Big).
\end{align}
Hence, the proposed V in \eqref{Eq10a} is an SPSF from  $\widehat \Sigma$ to $\Sigma$, which completes the proof. Note that the $\mathcal{K}$ and $\mathcal{K}_\infty$ functions $\kappa$, $\alpha$, and $\rho_{\mathrm{ext}}$, in Definition \ref{Def1} associated with the SPSF
in (\ref{Eq10a}) are $\alpha(s):=s^2$, $\kappa(s):=\widehat\kappa s$, and $\rho_{\mathrm{int}}(s):=(1+\frac{2}{\pi}+\frac{\pi}{2})\Vert\sqrt{M}D\Vert^2 s^2$, $\rho_{\mathrm{ext}}(s):=(1+\frac{2}{\pi}+\frac{2}{\pi})\Vert\sqrt{M}(B\widetilde R-P\hat B)\Vert^2 s^2$, $\forall s\in\R_{\ge0}$. Moreover, positive constant $\psi$ in \eqref{eq:V_dec} is $\psi=\text{Tr}\big(F^TMF+\hat F^TP^TMP\hat F\big)$.
\end{proof}

\begin{remark}
One can readily verify from the result of Theorem \ref{Thm3} that choosing $\hat F$ equal to zero results in smaller constant $\psi$ and, hence, more closeness of linear subsystems and their abstractions. Observe that this is not the case when one assumes the noise of the concrete subsystem and its abstraction are the same as in \cite{zamani2014compositional,zamani2016approximations}.
\end{remark}

\begin{remark}
Note that the results in Theorem \ref{Thm3} do not impose any condition on matrix $\hat B$ and, hence, it can be chosen arbitrarily. As an example, one
can choose $\hat B=I_{\hat n}$ which makes the abstract system $\widehat \Sigma$ fully
actuated and, hence, the synthesis problem over it much easier.
\end{remark}

\begin{remark}\label{Rem1}
Since Theorem \ref{Thm3} does not impose any condition on matrix $\widetilde R$, we choose $\widetilde R$ to minimize function $\rho_{\mathrm{ext}}$ for $V$ as suggested in \cite{girard2009hierarchical}. The following choice for $\widetilde R$
\begin{align}\label{Eq18a}
\widetilde R=(B^TMB)^{-1} B^TM P\hat B.
\end{align}
minimizes $\rho_{\mathrm{ext}}$.
\end{remark}
\section{Example}

Here, we demonstrate the effectiveness of the proposed results for an interconnected system consisting of four discrete-time linear stochastic control subsystems, i.e. $\Sigma=\mathcal{I}(\Sigma_1,\Sigma_2,\Sigma_3,\Sigma_4)$. The interconnection scheme of $\Sigma$ with four inputs and two outputs is illustrated in Figure \ref{Fig1}. 

As seen, the output of $\Sigma_1$ (resp. $\Sigma_2$) is connected to the internal input of
$\Sigma_{4}$ (resp. $\Sigma_3$) and the output of $\Sigma_{3}$ (resp. $\Sigma_4$) connects to the internal input of $\Sigma_{1}$ (resp. $\Sigma_2$). The system matrices are given by
%\begin{small}
\begin{align}\notag
A_i = I_{25}, ~B_i = I_{25}, ~C_i^T = 0.1\mathbf{1}_{25}, ~F_i = 0.01\mathbf{1}_{25},
\end{align}
for $i\in\{1,2,3,4\}$. The internal input and output matrices are also given by:
\begin{align}\notag
&C^T_{14}=C^T_{23}~=C^T_{31}=C^T_{42}= 0.1\mathbf{1}_{25},\\\notag
&D_{13}=	D_{24}=	D_{32}=	D_{41}= 0.1\mathbf{1}_{25}.
\end{align}
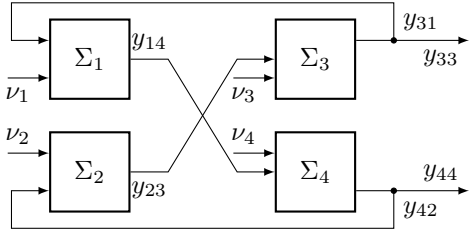
\begin{figure}
	\vspace{0.5cm}
	\centering
	\begin{tikzpicture}[auto, node distance=2cm, >=latex]
	\tikzset{block/.style    = {thick,draw, rectangle, minimum height = 3em, minimum width = 3em}}
	
	\node[block] (sys3) at (3,0) {$\Sigma_3$};
	\node[block] (sys4) at (3,-1.5) {$\Sigma_4$};
	\node[block] (sys1) at (0,0) {$\Sigma_1$};
	\node[block] (sys2) at (0,-1.5) {$\Sigma_2$};
	
	\draw[->] ($(sys3.east)+(0,0.25)$) -- node[near end,below] {$y_{33}$} ++(1.5,0);
	
	\draw[->] ($(sys4.east)+(0,-.25)$) -- node[near end, above] {$y_{44}$} ++(1.5,0);
	
	\draw[<-] ($(sys1.west)+(0,-.25)$) -- node[below, near end] {$\nu_{1}$} ++(-0.55,0);
	
	\draw[<-] ($(sys2.west)+(0,0.25)$) -- node[above, near end] {$\nu_{2}$} ++(-0.55,0);
	
	\draw[<-] ($(sys3.west)+(0,-.25)$) -- node[below, near end] {$\nu_{3}$} ++(-0.55,0);
	
	\draw[<-] ($(sys4.west)+(0,0.25)$) -- node[above, near end] {$\nu_{4}$} ++(-0.55,0);
	
	\draw[->] ($(sys3.east)+(.5,.25)$) |- node[near start, right] {$y_{31}$}  ($(sys1.west)+(-.5,.75)$) |- ($(sys1.west)+(0,.25)$) ;
	\draw[fill] ($(sys3.east)+(.5,.25)$) circle (1pt);
	
	\draw[->] ($(sys4.east)+(.5,-.25)$) |-node[near start, right] {$y_{42}$} ($(sys2.west)+(-.5,-.75)$) |- ($(sys2.west)+(0,-.25)$);
	\draw[fill] ($(sys4.east)+(.5,-.25)$) circle (1pt);
	
	\draw[->] ($(sys1.east)$)              -- node[above] {$y_{14}$} 
	($(sys1.east)+(0.5,0.00)$)  --
	($(sys4.west)+(-.5,0)$) -- (sys4.west);
	
	\draw[->] ($(sys2.east)$)              -- node[below] {$y_{23}$} 
	($(sys2.east)+(0.5,0.00)$)  --
	($(sys3.west)+(-.5,0)$) -- (sys3.west);
	
	\end{tikzpicture}
	\caption{The interconnected system $\Sigma=\mathcal{I}(\Sigma_1,\Sigma_2,\Sigma_3,\Sigma_4)$.}\label{Fig1}
	\vspace{-0.3cm}
\end{figure}
In order to construct an abstraction for $\mathcal{I}(\Sigma_1,\Sigma_2,\Sigma_3,\Sigma_4)$, we construct an abstraction $\widehat \Sigma_i$ of each individual subsystem $\Sigma_i$, $i\in\{1,2,3,4\}$. We first fix $\widehat\kappa$ and $\pi$ for each subsystem, and then determine the matrices $M$ and $K$ such that (\ref{Eq8a}) and (\ref{Eq9a}) hold for $i\in\{1,2,3,4\}$:
\begin{align}\notag
M_i = I_{25}, ~K_i = -0.95I_{25}, ~\widehat\kappa_i= 0.98, ~\pi_i = 0.99.
\end{align}
We continue with determining other matrices such that (\ref{Eq11a}), (\ref{Eq12a}), and (\ref{Eq14a}) hold:
\begin{align}\notag
P_i = \mathbf{1}_{25}, ~Q_i =\mathbf{1}_{25}, ~S_i= -0.003\mathbf{1}_{25},
\end{align}
for $i\in\{1,2,3,4\}$. Accordingly, the matrices of abstract subsystems are computed as:
\begin{align}\notag
\hat A_i = 2, ~\hat C_i =2.5, ~\hat D_i= 0.096,
\end{align}
for $i\in\{1,2,3,4\}$. Note that here $\hat F_i$, $i\in\{1,2,3,4\}$, are considered zero in order to reduce constants $\psi_i$ for each $V_i$. Moreover, $\hat B_i$ are chosen $1$ and we compute $\widetilde R_i$, $i\in\{1,2,3,4\}$, using (\ref{Eq18a}) as $\widetilde R_i = \mathbf{1}_{25}$.
The interface function for $i\in\{1,2,3,4\}$ follows by \eqref{Eq17aa} as:
\begin{align}\notag
\nu_i=-0.95I_{25}(x_i-\mathbf{1}_{25}\hat x_i)\!+\!\mathbf{1}_{25}\hat x_i\!+\!\mathbf{1}_{25}\hat \nu_i\!-\!0.003\mathbf{1}_{25}\hat \omega_i.
\end{align}
Hence, Theorem~\ref{Thm3} holds and ~$V_i(x_i,\hat x_i)=(x_i-\mathbf{1}_{25}\hat x_i)^T M_i(x_i-\mathbf{1}_{25}\hat x_i)$ is an SPSF function from~$\widehat\Sigma_i$ to~$\Sigma_i$ satisfying conditions \eqref{eq:V_dec1} and \eqref{eq:V_dec}  with $\alpha_i(s) = s^2\!, ~\kappa_i(s)=0.98s, ~\rho_{i\mathrm{ext}}(s)= 0, ~\rho_{i\mathrm{int}}(s)=0.88s^2\!,$ and $\psi_i=0.0025,$ for $i\in\{1,2,3,4\}$. We now proceed with Theorem \ref{Thm5} to construct a stochastic simulation function form $\widehat \Sigma$ to $\Sigma$. Assumption \ref{Asum1} holds with $\gamma_i(s)=s$ and:
\begin{equation*}
\resizebox{\hsize}{!}{%
	$
	\Delta=
	\begin{bmatrix}
	0 & 0 & 0.88 & 0\\
	0 & 0 & 0 & 0.88\\
	0 & 0.88& 0 & 0\\
	0.88& 0 & 0 & 0
	\end{bmatrix}\!\!,
	~\Lambda=
	\begin{bmatrix}
	0.98 & 0 & 0 & 0\\
	0 & 0.98 & 0 & 0\\
	0 & 0& 0.98 & 0\\
	0& 0 & 0 &0.98
	\end{bmatrix}\!\!.
	$}
\end{equation*} 
Additionally, one can readily verify that  a vector $\mu\in\R_{>0}^4$ exists here since the spectral radius of $\Lambda^{-1}\Delta$ is strictly less than one \cite{dashkovskiy2011small}. By choosing vector $\mu$ as $\mu^T=[1~\;1~\;1~\;1]$, the function
\begin{align}\notag
V(x,\hat x)\!=\!V_1(x_1,\hat x_1) \!\!+ \!\!V_2(x_2,\hat x_2)\notag
\!\!+\!\!V_3(x_3,\hat x_3) \!\!+ \!\!V_4(x_4,\hat x_4)
\end{align}
is an SSF from $\mathcal{I}(\widehat \Sigma_1, \widehat \Sigma_2,\widehat \Sigma_3,\widehat \Sigma_4)$ to
$\mathcal{I}(\Sigma_1,\Sigma_2,\Sigma_3,\Sigma_4)$ satisfying conditions \eqref{eq:lowerbound2} and \eqref{eq:martingale2}  with $\alpha(s)=s^2$, $\kappa(s)= 0.1s$, $\rho_{\mathrm{ext}}(s)=0$, $\forall s\in\R_{\ge0}$, and $\psi=0.01$. If the initial
states of the interconnected systems $\Sigma$ and $\widehat \Sigma$ are started from zero, one can readily verify that the norm of error between outputs of $\Sigma$ and of $\widehat \Sigma$ will not exceed $1$ with probability at least $90\%$ computed by the stochastic simulation function $V$ using inequality (\ref{Eq2a}) for $T=10$.

\section{Discussion}
In this paper, we provided a compositional approach for abstractions of interconnected discrete-time stochastic control systems, with independent noises in the abstract and concrete subsystems. First, we introduced new notions of stochastic pseudo-simulation and stochastic simulation functions in order to quantify the distance in a probability setting between original stochastic control subsystems and their abstractions and their interconnections, respectively. Therefore, one can employ the proposed results here to potentially solve safety/reachability problems over the abstract interconnected systems and then refine the results to the concrete interconnected ones. Furthermore,
we provided a computational scheme for the class of discrete-time linear stochastic control systems to construct abstractions together with their corresponding stochastic pseudo-simulation functions. Finally, we demonstrated the effectiveness of the results by constructing an
abstraction (totally $4$ dimensions) of the interconnection of four discrete-time linear stochastic control subsystems (together $100$ dimensions) in a compositional fashion.

\bibliographystyle{ieeetran}
\bibliography{biblio}

% Generated by IEEEtran.bst, version: 1.14 (2015/08/26)
\begin{thebibliography}{10}
\providecommand{\url}[1]{#1}
\csname url@samestyle\endcsname
\providecommand{\newblock}{\relax}
\providecommand{\bibinfo}[2]{#2}
\providecommand{\BIBentrySTDinterwordspacing}{\spaceskip=0pt\relax}
\providecommand{\BIBentryALTinterwordstretchfactor}{4}
\providecommand{\BIBentryALTinterwordspacing}{\spaceskip=\fontdimen2\font plus
\BIBentryALTinterwordstretchfactor\fontdimen3\font minus
  \fontdimen4\font\relax}
\providecommand{\BIBforeignlanguage}[2]{{%
\expandafter\ifx\csname l@#1\endcsname\relax
\typeout{** WARNING: IEEEtran.bst: No hyphenation pattern has been}%
\typeout{** loaded for the language `#1'. Using the pattern for}%
\typeout{** the default language instead.}%
\else
\language=\csname l@#1\endcsname
\fi
#2}}
\providecommand{\BIBdecl}{\relax}
\BIBdecl

\bibitem{julius2009approximations}
A.~A. Julius and G.~J. Pappas, ``Approximations of stochastic hybrid systems,''
  \emph{IEEE Transactions on Automatic Control}, vol.~54, no.~6, pp.
  1193--1203, 2009.

\bibitem{tkachev2011infinite}
I.~Tkachev and A.~Abate, ``On infinite-horizon probabilistic properties and
  stochastic bisimulation functions,'' in \emph{Proceedings of the 50th IEEE
  Conference on Decision and Control and European Control Conference
  (CDC-ECC)}, 2011, pp. 526--531.

\bibitem{SA13}
S.~{Esmaeil Zadeh Soudjani} and A.~Abate, ``Adaptive and sequential gridding
  procedures for the abstraction and verification of stochastic processes,''
  \emph{SIAM Journal on Applied Dynamical Systems}, vol.~12, no.~2, pp.
  921--956, 2013.

\bibitem{SSoudjani}
S.~{Esmaeil Zadeh Soudjani}, ``Formal abstractions for automated verification
  and synthesis of stochastic systems,'' Ph.D. dissertation, Technische
  Universiteit Delft, The Netherlands, November 2014.

\bibitem{zamani2014symbolic}
M.~Zamani, P.~Mohajerin~Esfahani, R.~Majumdar, A.~Abate, and J.~Lygeros,
  ``Symbolic control of stochastic systems via approximately bisimilar finite
  abstractions,'' \emph{IEEE Transactions on Automatic Control}, vol.~59,
  no.~12, pp. 3135--3150, 2014.

\bibitem{ZTA1}
M.~Zamani, I.~Tkachev, and A.~Abate, ``Towards scalable synthesis of stochastic
  control systems,'' \emph{Discrete Event Dynamic Systems}, vol.~27, no.~2, pp.
  341--369, July 2017.

\bibitem{zamani2014approximately}
M.~Zamani and A.~Abate, ``Approximately bisimilar symbolic models for randomly
  switched stochastic systems,'' \emph{Systems \& Control Letters}, vol.~69,
  pp. 38--46, 2014.

\bibitem{zamani2015symbolic}
M.~Zamani, A.~Abate, and A.~Girard, ``Symbolic models for stochastic switched
  systems: A discretization and a discretization-free approach,''
  \emph{Automatica}, vol.~55, pp. 183--196, 2015.

\bibitem{zamani2014compositional}
M.~Zamani, ``Compositional approximations of interconnected stochastic hybrid
  systems,'' in \emph{Proceedings of the 53rd IEEE Conference on Decision and
  Control (CDC)}, 2014, pp. 3395--3400.

\bibitem{zamani2016approximations}
M.~Zamani, M.~Rungger, and P.~Mohajerin~Esfahani, ``Approximations of
  stochastic hybrid systems: a compositional approach,'' \emph{IEEE
  Transactions on Automatic Control}, 2016.

\bibitem{hll1996}
O.~Hern{\'a}ndez-Lerma and J.~B. Lasserre, \emph{Discrete-time {M}arkov control
  processes}, ser. Applications of Mathematics.\hskip 1em plus 0.5em minus
  0.4em\relax Springer-Verlag, 1996, vol.~30.

\bibitem{BS96}
D.~P. Bertsekas and S.~E. Shreve, \emph{Stochastic {O}ptimal {C}ontrol: {T}he
  {D}iscrete-{T}ime {C}ase}.\hskip 1em plus 0.5em minus 0.4em\relax Athena
  Scientific, 1996.

\bibitem{1967stochastic}
H.~J. Kushner, \emph{Stochastic Stability and Control}, ser. Mathematics in
  Science and Engineering.\hskip 1em plus 0.5em minus 0.4em\relax Elsevier
  Science, 1967.

\bibitem{oksendal2013stochastic}
B.~Oksendal, \emph{Stochastic differential equations: an introduction with
  applications}.\hskip 1em plus 0.5em minus 0.4em\relax Springer Science \&
  Business Media, 2013.

\bibitem{young1912classes}
W.~H. Young, ``On classes of summable functions and their fourier series,''
  \emph{Proceedings of the Royal Society of London A: Mathematical, Physical
  and Engineering Sciences}, vol.~87, no. 594, pp. 225--229, 1912.

\bibitem{girard2009hierarchical}
A.~Girard and G.~J. Pappas, ``Hierarchical control system design using
  approximate simulation,'' \emph{Automatica}, vol.~45, no.~2, pp. 566--571,
  2009.

\bibitem{dashkovskiy2011small}
S.~Dashkovskiy, H.~Ito, and F.~Wirth, ``On a small gain theorem for iss
  networks in dissipative lyapunov form,'' \emph{European Journal of Control},
  vol.~17, no.~4, pp. 357--365, 2011.

\end{thebibliography}

\end{document}